\documentclass[sigconf]{acmart}

\usepackage{booktabs} 
\usepackage{epstopdf}
\usepackage{algorithm}
\usepackage{tabularx}
\usepackage{multirow}
\pdfmapfile{+txfonts.map}
\usepackage{caption}
\usepackage{subcaption}
\usepackage{algpseudocode}
\usepackage{amssymb}
\newtheorem{prop}{Proposition}

\AtBeginDocument{%
  \providecommand\BibTeX{{%
    \normalfont B\kern-0.5em{\scshape i\kern-0.25em b}\kern-0.8em\TeX}}}

\acmConference[WSDM '20]{ACM International Conference on Web Search and Data Mining}{Feb, 2020}{Houston, USA}
\acmBooktitle{WSDM '20: ACM International Conference on Web Search and Data Mining,  Feb, 2020, Houston, USA}
\acmPrice{15.00}
\acmISBN{978-1-4503-9999-9/18/06}

\begin{document}

\title{DDTCDR: Deep Dual Transfer Cross Domain Recommendation}

\author{Pan Li}
\affiliation{%
  \institution{New York University}
  \city{44th West 4th Street}
  \state{NY}
  \country{US}}
\email{pli2@stern.nyu.edu}
  
\author{Alexander Tuzhilin}
\affiliation{%
  \institution{New York University}
  \city{44th West 4th Street}
  \state{NY}
  \country{US}}
\email{atuzhili@stern.nyu.edu}

\begin{abstract}
Cross domain recommender systems have been increasingly valuable for helping consumers identify the most satisfying items from different categories. However, previously proposed cross-domain models did not take into account bidirectional latent relations between users and items. In addition, they do not explicitly model information of user and item features, while utilizing only user ratings information for recommendations. To address these concerns, in this paper we propose a novel approach to cross-domain recommendations based on the mechanism of dual learning that transfers information between two related domains in an iterative manner until the learning process stabilizes. We develop a novel latent orthogonal mapping to extract user preferences over multiple domains while preserving relations between users across different latent spaces. Combining with autoencoder approach to extract the latent essence of feature information, we propose Deep Dual Transfer Cross Domain Recommendation (DDTCDR) model to provide recommendations in respective domains. We test the proposed method on a large dataset containing three domains of movies, book and music items and demonstrate that it consistently and significantly outperforms several state-of-the-art baselines and also classical transfer learning approaches.
\end{abstract}


\maketitle

\section{Introduction}
Recommender systems have become the key component for online marketplaces to achieve great success in understanding user preferences. Collaborative filtering (CF) \cite{sarwar2001item} approaches, especially matrix factorization (MF) \cite{koren2009matrix} methods constitute the corner stone to help user discover satisfying products. However, these models suffer the cold start and the data sparsity problems \cite{schein2002methods,adomavicius2005toward}, for we only have access to a small fraction of past transactions, which makes it hard to model user preferences accurately and efficiently.

To address these problems, researchers propose to use cross domain recommendation \cite{fernandez2012cross} through transfer learning \cite{pan2010survey} approaches that learn user preferences in the source domain and transfer them to the target domain. For instance, if a user watches a certain movie, we will recommend the original novel on which the movie is based to that user. Most of the methods along this line focus on the unidirectional transfer learning from the source domain to the target domain and achieve great recommendation performance \cite{fernandez2012cross}. In addition, it is beneficial for us to also transfer user preferences in the other direction via dual transfer learning \cite{long2012dual,zhong2009cross,wang2011cross}. For example, once we know the type of books that the user would like to read, we can recommend movies on related topics to form a loop for better recommendations in both domains \cite{li2009can}. 

However, previous dual transfer models only focus on explicit information of users and items, without considering latent and complex relations between them. Besides, they do not include user and item features during the recommendation process, thus limiting the recommender system to provide satisfying results. Latent embedding methods, on the other hand, constitute a powerful tool to extract latent user preferences from the data record and model user and item features efficiently \cite{zhang2019deep}. Therefore, it is crucially important to design a model that utilize latent embedding approach to conduct dual transfer learning for cross domain recommendations.

In this paper, we apply the idea of deep dual transfer learning to cross-domain recommendations using latent embeddings of user and item features. We assume that if two users have similar preferences in a certain domain, their preferences should also be similar across other domains as well. We address this assumption by proposing a unifying mechanism that extracts the essence of preference information in each domain and improves the recommendations for both domains simultaneously through better understanding of user preferences. 

Specifically, we propose Deep Dual Transfer Cross Domain Recommendation (DDTCDR) model that learns latent orthogonal mappings across domains and provides cross domain recommendations leveraging user preferences from all domains. Furthermore, we empirically demonstrate that by iteratively updating the dual recommendation model, we simultaneously improve recommendation performance over both domains and surpass all baseline models. Compared to previous approaches, the proposed model has the following advantages:
\begin{itemize}
\item It transfers latent representations of features and user preferences instead of explicit information between the source domain and the target domain to capture latent and complex interactions as well as modeling feature information.
\item It utilizes deep dual transfer learning mechanism to enable bidirectional transfer of user preferences that improves recommendation performance in both domains simultaneously over time.
\item It learns the latent \textit{orthogonal} mapping function across the two domains, that is capable of preserving similarities of user preferences and computing the inverse mapping function efficiently.
\end{itemize}

In this paper, we make the following contributions:
\begin{itemize}
\item We propose to apply the combination of dual transfer learning mechanism and latent embedding approach to the problem of cross domain recommendations.
\item We empirically demonstrate that the proposed model outperforms the state-of-the-art approaches and improves recommendation accuracy across multiple domains and experimental settings.
\item We theoretically demonstrate the convergence condition for the simplified case of our model and empirically show that the proposed model stabilizes and converges after several iterations.
\item We illustrate that the proposed model can be easily extended for multiple-domain recommendation applications.
\end{itemize}


\section{Related Work}
The proposed model stems from two research directions: cross domain recommendations and deep-learning based recommendations. Also, we discuss the literature on dual transfer learning, and how they motivate the combination of dual transfer learning mechanism and latent embedding methods for cross domain recommendations.

\subsection{Cross Domain and Transfer Learning-based Recommendations}
Cross domain recommendation approach \cite{fernandez2012cross} constitutes a powerful tool to deal with the data sparsity problem. Typical cross domain recommendation models are extended from single-domain recommendation models, including CMF \cite{singh2008relational}, CDCF \cite{li2009can,hu2013personalized}, CDFM \cite{loni2014cross}, Canonical Correlation Analysis \cite{sahebi2017cross,sahebi2015takes,sahebi2014content,sahebi2013cross} and Dual Regularization \cite{wu2018dual}.  These approaches assume that different patterns characterize the way that users interact with items of a certain domain and allow interaction information from an auxiliary domain to inform recommendation in a target domain. 

The idea of information fusion also motivates the use of transfer learning \cite{pan2010survey} that transfers extracted information from the source domain to the target domain.  Specifically, researchers \cite{hu2018conet,lian2017cccfnet} propose to learn the user preference in the source domain and transfer the preference information into target domain for better understanding of user preferences. These models have achieved great success in addressing the cold-start problem and enhancing recommendation performance. 

However, these models do not fundamentally address the relationship between different domains, for they do not improve recommendation performance of both domains simultaneously, thus might not release the full potential of utilizing the cross-domain user interaction information. Also they do not explicitly model user and item features during recommendation process. In this paper, we propose to use a novel dual transfer learning mechanism combining with autoencoder to overcome these issues and significantly improve recommendation performance.

\subsection{Dual Transfer Learning}
Transfer learning \cite{pan2010survey} deals with the situation where the data obtained from different resources are distributed differently. It assumes the existence of common knowledge structure that defines the domain relatedness, and incorporate this structure in the learning process by discovering a shared latent feature space in which the data distributions across domains are close to each other. The existing transfer learning methods for cross-domain recommendation include Collaborative DualPLSA \cite{zhuang2010collaborative}, Joint Subspace Nonnegative Matrix Factorization \cite{liu2013multi}  and JDA \cite{long2013transfer} that learn the latent factors and associations spanning a shared subspace where the marginal distributions across domains are close.

In addition, to exploit the duality between these two distributions and to enhance the transfer capability, researchers propose the dual transfer learning mechanism \cite{long2012dual,zhong2009cross,wang2011cross} that simultaneously learns the marginal and conditional distributions. Recently, researchers manage to achieve great performance on machine translation with dual-learning mechanism \cite{he2016dual,xia2017dual,wang2018dual}. All these successful applications address the importance of exploiting the duality for mutual reinforcement. However, none of them apply the dual transfer learning mechanism into cross-domain recommendation problems, where the duality lies in the symmetrical correlation between source domain and target domain user preferences. In this paper, we utilize a novel dual-learning mechanism and significantly improve recommendation performance.

\subsection{Deep Learning-based Recommendations}
Recently, deep learning has been revolutionizing the recommendation architectures dramatically and brings more opportunities to improve the performance of existing recommender systems. To capture the latent relationship between users and items, researchers propose to use deep learning based recommender systems \cite{zhang2019deep,wang2015collaborative}, especially embedding methods \cite{he2017neural} and autoencoding methods \cite{sedhain2015autorec,wu2016collaborative,li2017collaborative,liang2018variational} to extract the latent essence of user-item interactions for the better understanding of user preferences.

However, user preferences in different domains are learned separately without exploiting the duality for mutual reinforcement, for researchers have not addressed the combination of deep learning methods with dual transfer learning mechanism in recommender systems, which learns user preferences from different domains simultaneously and further improves the recommendation performance. To this end, we propose a dual transfer collaborative filtering model that captures latent interactions across different domains. The effectiveness of dual transfer learning over the existing methods is demonstrated by extensive experimental evaluation.

\begin{figure*}[!]
\centering
\includegraphics[width=\textwidth]{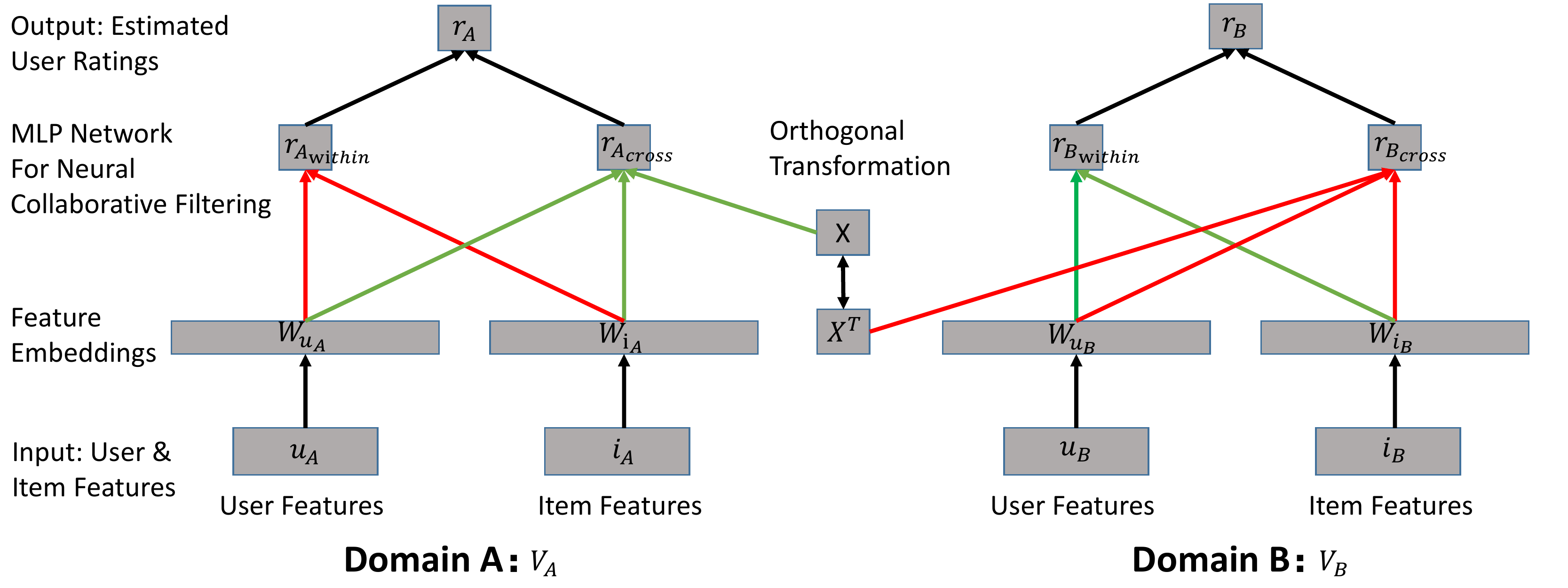}
\caption{Model Framework: Red and blue lines represent the recommendation model for domain A and B respectively. We obtain the estimated ratings by taking the linear combination of within-domain and cross-domain user preferences and backpropogate the loss to update the two models and orthogonal mappings simultaneously.}
\label{framework}
\end{figure*}

\begin{algorithm*}[!]
\caption{Dual Neural Collaborative Filtering}\label{alg}
\begin{algorithmic}[1]
\State \textbf{Input}: Domain $V_A$ and $V_B$, autoencoder $AE_{A}$ and $AE_{B}$, transfer rate $\alpha$, learning rates $\gamma_{A}$ and $\gamma_{B}$,  initial recommendation models $RS_{A}$ and $RS_{B}$, initial mapping function X
\Repeat
\State Sample user-item records $d_A$ and $d_B$ from $V_A$ and $ V_B$ respectively
\State Unpack records $d_A, d_B$ as user features $u_A, u_B$, item features $i_A, i_B$ and ratings $r_A, r_B$
\State Generate feature embeddings from autoencoder as $W_{u_{A}} = AE_{A}(u_{A})$, $W_{u_{B}} = AE_{B}(u_{B})$, $W_{i_{A}} = AE_{A}(i_{A})$, $W_{i_{B}} = AE_{B}(i_{B})$
\State Estimate the ratings in domain A via $r_A'= (1-\alpha)RS_{A}(W_{u_{A}},W_{i_{A}}) + \alpha RS_{B}(X*W_{u_{A}},W_{i_{A}})$
\State Estimate the ratings in domain B via $r_B'= (1-\alpha)RS_{B}(W_{u_{B}},W_{i_{B}}) + \alpha RS_{A}(X^T*W_{u_{B}},W_{i_{B}})$
\State Compute MSE loss $\hat{r_{A}} = r_A - r_A'$, $\hat{r_{B}} = r_B - r_B'$
\State Backpropogate $\hat{r_{A}}$, $\hat{r_{B}}$ and update $RS_{A}$, $RS_{B}$;
\State Backpropogate orthogonal constraint on $X$;  Orthogonalize $X$
\Until{convergence}
\end{algorithmic}
\label{dualncf}
\end{algorithm*}

\section{Method}
In this section, we present the proposed Deep Dual Transfer Cross Domain Recommendation (DDTCDR) model. First, we construct feature embeddings using the autoencoding technique from the user and item features. Then we design a specific type of latent orthogonal mapping function for transferring feature embeddings across two different domains through dual transfer learning mechanism. Moreover, we theoretically demonstrate convergence of the proposed model under certain assumptions. The complete architecture of DDTCDR system is illustrated in Figure \ref{framework}. As shown in Figure \ref{framework}, we take the user and item features as input and map them into the feature embeddings, from which we compute the within-domain and cross-domain user preferences and provide recommendations. Important mathematical notations used in our model are listed in Table \ref{notation}. We explain the details in the following sections.

\begin{table}[b!]
\centering
\caption{Mathematical Notations.}\label{notation}
\vspace{-2mm}
\begin{tabular}{cc}
\hline
\textbf{Symbol} & \textbf{Description} \\ \hline
$u$ & User Features \\
$i$ & Item Features \\
$W_{u}$ & User Feature Embeddings \\
$W_{i}$ & Item Feature Embeddings \\
$\gamma$ & Learning Rate \\
$r$ & Estimated Ratings \\
$r_{within}$ & Within-Domain Estimated Ratings\\
$r_{cross}$ & Cross-Domain Estimated Ratings \\
$X$ & Latent Orthogonal Mapping\\
$X^T$ & Transpose of Latent Orthogonal Mapping \\
$RS$ & Domain-Specific Recommender System \\
$AE$ & Domain-Specific Autoencoder \\
$\alpha$ & Hyperparameter in Hybrid Utility Function \\
\hline
\end{tabular}
\vspace{-0mm}
\end{table}

\subsection{Feature Embeddings}
To effectively extract latent user preferences and efficiently model features of users and items, we present an autoencoder framework that learns the latent representations of user and item features and transforms the heterogeneous and discrete feature vectors into continuous feature embeddings. We denote the feature information for user $a$ as $u_{a} = \{u_{a_{1}},eu_{a_{2}},\cdots,u_{a_{m}}\}$ and the feature information for item $b$ as $i_{b} = \{i_{b_{1}},i_{b_{2}},\cdots,i_{b_{n}}\}$, where $m$ and $n$ stand for the dimensionality of user and item feature vectors respectively. The goal is to train two separate neural networks: encoder that maps feature vectors into latent embeddings, and decoder that reconstructs feature vectors from latent embeddings. Due to effectiveness and efficiency of the training process, we formulate both the encoder and the decoder as multi-layer perceptron (MLP). MLP learns the hidden representations by optimization reconstruction loss $L$:
\begin{equation}
L = ||u_{a}-MLP_{dec}(MLP_{enc}(u_{a}))||
\end{equation}
where $MLP_{enc}$ and $MLP_{dec}$ represents the MLP network for encoder and decoder respectively. Note that, in this step we train the autoencoder separately for users and items in different domains to avoid the information leak between domains.

\subsection{Latent Orthogonal Mapping}
In this section we introduce the latent orthogonal mapping function for transferring user preferences from the source domain to the target domain. The fundamental assumption of our proposed approach is that, if two users have similar preferences in the source domain, their preferences should also be similar in the target domain. This implies that we need to preserve similarities between users during transfer learning across domains. In particular, we propose to use a latent orthogonal matrix to transfer information across domains for two reasons. First, it preserves similarities between user embeddings across different latent spaces since orthogonal transformation preserves inner product of vectors. Second, it automatically derives the inverse mapping matrix as $X^T$, because $Y = X^T(XY)$ holds for any given orthogonal mapping matrix $X$, thus making inverse mapping matrix equivalent to its transpose. This simplifies the learning procedure and reduces the perplexity of the recommendation model.

Using the latent orthogonal mapping, we could transfer user preferences from one domain to the other. In the next section, we will introduce model to utilize this latent orthogonal mapping for constructing dual learning mechanism for cross domain recommendations.

\subsection{Deep Dual Transfer Learning}
As an important tool to provide recommendations, matrix factorization methods associate user-item pairs with a shared latent space, and use the latent feature vectors to represent users and items. In order to model latent interactions between users and items as well as feature information, it is natural to generalize matrix factorization methods using latent embedding approaches. However, in the cross domain recommendation application, we may not achieve the optimal performance by uniformly applying neural network model to all domains of users and items due to the cold-start and sparsity problems.

To address these problems and improve performance of recommendations, we propose to combine dual transfer learning mechanism with cross domain recommendations, where we transfer user preferences between the source domain and the target domain simultaneously. Consider two different domains $D_{A}$ and $D_{B}$ that contain user-item interactions as well as user and item features. In real business practices, user group in $D_{A}$ often overlaps with $D_{B}$, meaning that they have purchased items in both domains, while there is no overlap of item between two domains for each item only belongs to one single domain. One crucial assumption we make here is that, if two users have similar preferences in $D_{A}$, they are supposed to have similar preferences in $D_{B}$ as well, indicating that we can learn and improve user preferences in both domains simultaneously over time. Therefore to obtain better understanding of user preferences in $D_{A}$, we also utilize external user preference information from $D_{B}$ and combine them together. Similarly, we could get better recommendation performance in $D_{B}$ if we utilize user preference information from $D_{A}$ at the same step. To leverage duality of the two transfer learning based recommendation models and to improve effectiveness of both tasks simultaneously, we conduct dual transfer learning recommendations for the two models \textit{together} and learn the latent mapping function accordingly.  

Specifically, we propose to model user preferences $r$ using two components: \textit{within-domain preference} $r_{within}$ that captures user interactions and predict user behaviors in the target domain and \textit{cross-domain preference} $r_{cross}$ that utilizes user actions from the source domain. We also introduce transfer rate $\alpha$ as a hyperparameter, which represents the relative importance of the two components in prediction of user preferences. We propose to estimate user ratings in domain pairs $(A,B)$ as follows:
\begin{equation}
r_A'= (1-\alpha)RS_{A}(W_{u_{A}},W_{i_{A}}) + \alpha RS_{B}(X*W_{u_{A}},W_{i_{A}})
\end{equation}
\begin{equation}
r_B'= (1-\alpha)RS_{B}(W_{u_{B}},W_{i_{B}}) + \alpha RS_{A}(X^T*W_{u_{B}},W_{i_{B}})
\end{equation}
where $W_{u_{A}},W_{i_{A}},W_{u_{B}},W_{i_{B}}$ represents user and item embeddings and $RS_{A}, RS_{B}$ stand for the neural recommendation model for domain A and B respectively. The first term in each equation computes within-domain user preferences from user and item features in the same domain, while the second term denotes cross-domain user preferences obtained by the latent orthogonal mapping function $X$ to capture heterogeneity of different domains. We use weighted linear combination for rating estimations in equations (4) and (5). When $\alpha = 0$ or there is no user overlap between two domains, the dual learning model degenerates to two separate single-domain recommendation models; when $\alpha = 0.5$ and $W = \textbf{I}$, the dual learning model degenerates to one universal recommendation model across two domains. For users that only appear in one domain, the hyperparameter $\alpha$ is selected as 0. For users that appear in both domains, $\alpha$ normally should take a positive value between 0 and 0.2 based on the experimental results conducted in Section 6, which indicates that within-domain preferences play the major role in the understanding of user behavior.

As shown in Figure \ref{framework} and equations (4) and (5), dual transfer learning entails the transfer loop across two domains and the learning process goes through the loop iteratively. It is important to study the convergence property of the model, which we discuss in the next section.

\subsection{Convergence Analysis}
In this section, we present the convergence theorem of matrix factorization with dual transfer learning mechanism. We denote the rating matrices as $V_{A},V_{B}$ for the two domains A and B respectively. The goal is to find the approximate non-negative factorization matrices that simultaneously minimizes the reconstruction loss. We conduct multiplicative update algorithms \cite{lin2007convergence} for the two reconstruction loss iteratively given $\alpha$, $X$, $V_{A}$ and $V_{B}$.
\begin{equation}
\begin{split}
\displaystyle \min_{W_{A},W_{B},H_{A},H_{B}} & |V_{A} - (1-\alpha) W_{A}H_{A} - \alpha XW_{B}H_{B}| \\
  + & |V_{B} - (1-\alpha) W_{B}H_{B} - \alpha X^TW_{A}H_{A}|
\end{split}
\end{equation}
\begin{prop}
Convergence of the iterative optimization of dual matrix factorization for rating matrices $V_{A}$ and $V_{B}$ in (6) is guaranteed.
\end{prop}
\begin{proof}
Combining the two parts of objective functions and eliminating $W_{A}H_{A}$ term, we will get
\begin{equation}
(1-\alpha) V_{B} - \alpha X^TV_{A} = (1- 2\alpha)W_{B}H_{B}
\end{equation}
Similarly, when we eliminate $W_{B}H_{B}$ term, we will get
\begin{equation}
(1-\alpha) V_{A} - \alpha XV_{B} = (1- 2\alpha)W_{A}H_{A}
\end{equation}
Based on the analysis in \cite{lee2001algorithms} for the classical single-domain matrix factorization, to show that the repeated iteration of update rules is guaranteed to converge to a locally optimal solution, it is sufficient to show that 
$$ \left\{
\begin{array}{rcl}
2\alpha - 1 < 0       &      &  (a)\\
(1-\alpha) V_{B} - \alpha X^TV_{A} \ge \textbf{0}    &      & (b)\\
(1-\alpha) V_{A} - \alpha XV_{B} \ge \textbf{0}     &      & (c)\\
\end{array} \right. $$
where \textbf{0} stands for the zero matrix. Condition (a) is an intuitive condition which indicates that the information from the target domain dominantly determines the user preference, while information transferred from the source domain only serves as the regularizer in the learning period. To fulfill the seemingly complicated condition (b) and (c), we recall that $V_{A}$ and $V_{B}$ are ratings matrices, which are non-negative and bounded by the rating scale $k$. We design two ''positive perturbation" $V_{A}' = V_{A} + mk\textbf{1}$ and $V_{B}' = V_{B} + mk\textbf{1}$ where $m$ is the rank of the mapping matrix $X$. Condition (a) is independent of the specific rating matrices, and we could check that $V_{A}'$ and $V_{B}'$ satisfies condition (b) and (c). Thus, the matrix factorization of $V_{A}'$ and $V_{B}'$ is guaranteed convergence; to reconstruct the original matrix $V_{A}$ and $V_{B}$, we only need to deduce the ''positive perturbation'' term $m$, so the original problem is also guaranteed convergence.
\end{proof}

In this paper, to capture latent interactions between users and items, we use the neural network based matrix factorization approach, therefore it still remains unclear if this convergence happens to the proposed model. However, our hypothesis is that even in our case we will experience similar convergence process guaranteed in the classical matrix factorization case, as stated in the proposition. We test the hypothesis in Section 6.

\subsection{Extension to Multiple Domains}
In previous sections, we describe the proposed DDTCDR model with special focus on the idea of combining dual transfer learning mechanism with latent embedding approaches. We point out that the proposed model not only works for cross domain recommendations between two domains, but can easily extend to recommendations between multiple domains as well. Consider $N$ different domains $D_{1}, D_{2}, D_{3}, \cdots, D_{n}$. To provide recommendations for domain $D_{k}$, we estimate the ratings similarly as the hybrid combination of within-domain estimation and cross-domain estimation:
\begin{equation}
r_k'= (1-\alpha)RS_{k}(W_{u_{k}},W_{i_{k}}) + \frac{\alpha}{n-1}\sum_{j=1;j\neq k}^{n} RS_{j}(X_{jk}*W_{u_{k}},W_{i_{k}})
\end{equation}
where $X_{jk}$ represents the latent orthogonal transfer matrix between domain $D_{j}$ and $D_{k}$ respectively. Therefore, the proposed model is capable of providing recommendations for multiple-domain application effectively and efficiently.

\section{Experiment}
To validate the performance of the proposed model, we compare the cross-domain recommendation accuracy with several state-of-the-art methods. In addition, we conduct multiple experiments to study sensitivity of hyperparameters of our proposed model.

\subsection{Dataset}
We use a large-scale anonymized dataset obtained from a European online recommendation service for carrying our experiments. It allows users to rate and review a range of items from various domains, while each domain was treated as an independent sub-site with separate within-domain recommendations. Consequently, the combination of explicit user feedback and diverse domains makes the dataset unique and valuable for cross-domain recommendations. We select the subset that includes three largest domains - books, movies and music, three domains being linked together through a common user ID identifying the same user. We normalize the scale of ratings between 0 and 1. Note that each user might purchase items from different domains, while each item only belongs to one single domain. The basic statistics about this dataset are shown in Table \ref{statisticalnumber}. 

\begin{table}
\centering
\begin{tabular}{l l l l}
\hline
Domain & Book & Movie & Music \\ \hline
\# Users & 804,285 & 959,502 & 45,962 \\ 
\# Items & 182,653 & 79,866 & 183,114 \\ 
\# Ratings & 223,007,805 & 51,269,130 & 2,536,273 \\ 
Density & 0.0157\% & 0.0669\% & 0.0301\% \\
\hline
\end{tabular}
\newline
\caption{Descriptive Statistics for the Dataset}
\label{statisticalnumber}
\end{table}

\subsection{User and Item Features}
Besides the records of interactions between users and items, the online platform also collects user information by asking them online questions when they are using the recommender services. Typical examples include questions about preferences between two items, ideas about recent events, life styles, demographic information, etc.. From all these questions, we select eight most-popular questions and use the corresponding answers from users as user features. Meanwhile, although the online platform does not collect directly item features, we obtain the metadata through Google Books API\footnote{https://developers.google.com/books/}, IMDB API\footnote{http://www.omdbapi.com/} and Spotify API\footnote{https://developer.spotify.com/documentation/web-api/}. To ensure the correctness of the collected item features, we validate the retrieved results with the timestamp included in the dataset. We describe the entire feature set used in our experiments in Table \ref{featureset}.

\begin{table}
\footnotesize
\centering
\begin{tabular}{l l c l}
\hline
Category & Feature Group & Dimensionality & Type \\ \hline
\multirow{8}{*}{User Features} & Gender & 2 & one-hot \\
                                                     & Age      & $\sim 10^2$ & numerical \\
                                                     & Movie Taste & 12 & one-hot \\
                                                     & Residence & 12 & one-hot \\
                                                     & Preferred Category & 9 & one-hot \\
                                                     & Recommendation Usage & 5 & one-hot \\
                                                     & Marital Status & 3 & one-hot \\
                                                     & Personality & 6 & one-hot \\ \hline
\multirow{8}{*}{Book Features}  & Category & 8 & one-hot \\
                                                      & Title         & $\sim 10^5$ & one-hot \\ 
                                                      & Author     & $\sim 10^4$ & multi-hot \\
                                                      & Publisher  & $\sim 10^2$ & one-hot  \\
                                                      & Language & 4 & one-hot \\
                                                      & Country    & 4 & one-hot \\
                                                      & Price         & $\sim 10^2$ & numeric \\
                                                      & Date         & $\sim 10^3$ & date \\ \hline
\multirow{8}{*}{Movie Features} & Genre       & 6 & one-hot \\
                                                      & Title          & $\sim 10^5$ & one-hot \\
                                                      & Director    & $\sim 10^3$ & multi-hot \\     
                                                      & Writer       & $\sim 10^3$ & multi-hot \\
                                                      & Runtime    & $\sim 10^3$ & numeric \\
                                                      & Country    & 4 & one-hot \\
                                                      & Rating       & $\sim 10^2$ & numeric \\
                                                      & Votes         & $\sim 10^4$ & numeric \\ \hline
\multirow{8}{*}{Music Features} & Listener     & $\sim 10^3$ & numeric \\
                                                      & Playcount & $\sim 10^3$ & numeric \\
                                                      & Artist         & $\sim 10^4$ & one-hot \\
                                                      & Album        & $\sim 10^4$ & one-hot \\
                                                      & Tag            & 8 & one-hot \\
                                                      & Release     & $\sim 10^3$ & date \\
                                                      & Duration    & $\sim 10^3$ & numeric \\
                                                      & Title           & $\sim 10^5$ & one-hot \\ \hline
\end{tabular}
\newline
\caption{Statistics of feature sets used in the purposed model}
\label{featureset}
\end{table}

\subsection{Baseline Models}
To conduct experiments and evaluations of our model, we utilize the record-stratified 5-fold cross validation and evaluate the recommendation performance based on RMSE, MAE, Precision and Recall metrics \cite{ricci2011introduction}. We compare the performance with a group of state-of-the-art methods.

\begin{itemize}
\item \textbf{CCFNet\cite{lian2017cccfnet}} Cross-domain Content-boosted Collaborative Filtering neural NETwork (CCCFNet) utilizes factorization to tie CF and content-based filtering together with a unified multi-view neural network.
\item \textbf{CDFM\cite{loni2014cross}} Cross Domain Factorization Machine (CDFM) proposes an extension of FMs that incorporates domain information in this pattern, which assumes that user interaction patterns differ sufficiently to make it advantageous to model domains separately. 
\item \textbf{CoNet\cite{hu2018conet}} Collaborative Cross Networks (CoNet) enables knowledge transfer across domains by cross connections between base networks.
\item \textbf{NCF\cite{he2017neural}} Neural Collaborative Filtering (NCF) is a neural network architecture to model latent features of users and items using collaborative filtering method. The NCF models are trained separately for each domain without transferring any information.
\item \textbf{CMF\cite{singh2008relational}}  Collective Matrix Factorization (CMF) simultaneously factor several matrices, sharing parameters among factors when a user participates in multiple domains.
\end{itemize}
 Furthermore, we select hyperparamter $\alpha$ for our study as 0.03 using Bayesian Optimization. We use one-layer MLP for constructing encoder and decoder separately. The size of feature embedding is 8. For each baseline method, we set the same hyperparameters as our proposed model if applicable. The results are reported in the next section.

\section{Results}

\subsection{Cross-Domain Recommendation Performance}
Since we have three different domains, i.e. books, movies and music, this results in three domain pairs for evaluation. The performance comparison results of DDTCDR with the baselines are reported using the experimental settings in Section 4. 

As shown in Table \ref{dual1}, \ref{dual2} and \ref{dual3}, the proposed DDTCDR model significantly and consistently outperforms all other baselines in terms of RMSE, MAE, Precision and Recall metrics across all the three domain pairs. As Table 3 shows, RMSE measures for book vs. movie domains obtained using the proposed model are 0.2213 and 0.2213, which outperform the second-best baselines by 3.98\% and 2.44\%; MAE measures are 0.1708 and 0.1704, which outperform the second-best baselines by 9.54\% and 9.80\%. We observe similar improvements in book vs. music domains where DDTCDR outperforms the second-best baselines by 4.07\%, 8.87\%, 2.14\%, 4.74\% and improvements in movie vs. music domains by 3.75\%, 9.77\%, 1.89\%, 4.24\% respectively. To summarize, all these results show that the dual transfer learning mechanism works well in practice, and the proposed DDTCDR model achieves significant cross-domain recommendation performance improvements. It is also worth noticing that DDTCDR only requires the same level of time complexity for training compared with other baseline models.

Furthermore, we observe that the improvements of proposed model on the book and movie domains are relatively greater than that of the music domain. One possible reason is the size of the dataset, for the book and movie domains are significantly larger than the music domain. We plan to explore this point further as a topic of future research.

\begin{table*}
\centering
\begin{tabular}{|c|cccc|cccc|} \hline
\multirow{2}{*}{Algorithm} & \multicolumn{4}{c|}{Book} & \multicolumn{4}{c|}{Movie} \\ \cline{2-9}
               &  RMSE & MAE & Precision@5 & Recall@5 & RMSE & MAE & Precision@5 & Recall@5 \\ \hline
\textbf{DDTCDR} & \textbf{0.2213*} & \textbf{0.1708*} & \textbf{0.8595*} & \textbf{0.9594*} & \textbf{0.2213*} & \textbf{0.1714*} & \textbf{0.8925*} & \textbf{0.9871*} \\
Improved \% & (+3.98\%) & (+9.54\%) & (+2.77\%) & (+6.30\%) & (+2.44\%) & (+9.80\%) & (+2.75\%) & (+2.74\%) \\ \hline
NCF & 0.2315 & 0.1887 & 0.8357 & 0.8924 & 0.2276 & 0.1895 & 0.8644 & 0.9589 \\ \hline
CCFNet & 0.2639 & 0.1841 & 0.8102 & 0.8872 & 0.2476 & 0.1939 & 0.8545 & 0.9300 \\ \hline
CDFM & 0.2494 & 0.2165 & 0.7978 & 0.8610 & 0.2289 & 0.1901 & 0.8498 & 0.9312 \\ \hline
CMF & 0.2921 & 0.2478 & 0.7972 & 0.8523 & 0.2738 & 0.2293 & 0.8324 & 0.9012 \\ \hline
CoNet & 0.2305 & 0.1892 & 0.8328 & 0.8990 & 0.2298 & 0.1903 & 0.8680 & 0.9601 \\ \hline
\end{tabular}
\newline
\caption{Comparison of recommendation performance in Book-Movie Dual Recommendation: Improved Percentage versus the second best baselines}
\label{dual1}
\end{table*}

\begin{table*}
\centering
\begin{tabular}{|c|cccc|cccc|} \hline
\multirow{2}{*}{Algorithm} & \multicolumn{4}{c|}{Book} & \multicolumn{4}{c|}{Music} \\ \cline{2-9}
               &  RMSE & MAE & Precision@5 & Recall@5 & RMSE & MAE & Precision@5 & Recall@5 \\ \hline
\textbf{DDTCDR} & \textbf{0.2209*} & \textbf{0.1704*} & \textbf{0.8570*} & \textbf{0.9602*} & \textbf{0.2753*} & \textbf{0.2302*} & \textbf{0.8392*} & \textbf{0.8928*} \\
Improved \% & (+4.07\%) & (+8.87\%) & (+3.97\%) & (+3.15\%) & (+2.14\%) & (+4.74\%) & (+5.51\%) & (+5.35\%) \\ \hline
NCF & 0.2315 & 0.1887 & 0.8230 & 0.9294 & 0.2828 & 0.2423 & 0.7930 & 0.8450 \\ \hline
CCFNet & 0.2630 & 0.1842 & 0.8150 & 0.9108 & 0.3090 & 0.2422 & 0.7902 & 0.8388 \\ \hline
CDFM & 0.2489 & 0.2155 & 0.8104 & 0.9102 & 0.3252 & 0.2463 & 0.7895 & 0.8365 \\ \hline
CMF & 0.2921 & 0.2478 & 0.8072 & 0.8978 & 0.3478 & 0.2698 & 0.7820 & 0.8324 \\ \hline
CoNet & 0.2307 & 0.1897 & 0.8230 & 0.9300 & 0.2801 & 0.2410 & 0.7912 & 0.8428 \\ \hline
\end{tabular}
\newline
\caption{Comparison of recommendation performance in Book-Music Dual Recommendation: Improved Percentage versus the second best baselines}
\label{dual2}
\end{table*}

\begin{table*}
\centering
\begin{tabular}{|c|cccc|cccc|} \hline
\multirow{2}{*}{Algorithm} & \multicolumn{4}{c|}{Movie} & \multicolumn{4}{c|}{Music} \\ \cline{2-9}
               &  RMSE & MAE & Precision@5 & Recall@5 & RMSE & MAE & Precision@5 & Recall@5 \\ \hline
\textbf{DDTCDR} & \textbf{0.2174*} & \textbf{0.1720*} & \textbf{0.8926*} & \textbf{0.9869*} & \textbf{0.2758*} & \textbf{0.2311*} & \textbf{0.8370*} & \textbf{0.8902*} \\
Improved \% & (+3.75\%) & (+9.77\%) & (+5.32\%) & (+3.68\%) & (+1.89\%) & (+4.24\%) & (+4.30\%) & (+4.38\%) \\ \hline
NCF &  0.2276 & 0.1895 & 0.8428 & 0.9495 & 0.2828 & 0.2423 & 0.7970 & 0.8501 \\ \hline
CCFNet & 0.2468 & 0.1932 & 0.8398 & 0.9310 & 0.3090 & 0.2433 & 0.7952 & 0.8498 \\ \hline
CDFM & 0.2289 & 0.1895 & 0.8306 & 0.9382 & 0.3252 & 0.2467 & 0.7880 & 0.8460 \\ \hline
CMF & 0.2738 & 0.2293 & 0.8278 & 0.9222 & 0.3478 & 0.2698  & 0.7796 & 0.8400 \\ \hline
CoNet & 0.2302 & 0.1908 & 0.8450 & 0.9508 & 0.2811 & 0.2428 & 0.8010 & 0.8512 \\ \hline
\end{tabular}
\newline
\caption{Comparison of recommendation performance in Movie-Music Dual Recommendation: Improved Percentage versus the second best baselines}
\label{dual3}
\end{table*}

\subsection{Convergence}
In Section 3.4, we proof convergence of the dual transfer learning method for classical matrix factorization problem. Note that however, this proposition is not applicable to the DDTCDR model because the optimization process is different. Therefore, we test the convergence empirically in our study, and conjecture should still happen even though our method does not satisfy the condition. We conduct the convergence study to empirically demonstrate the convergence. 

In particular, we train the model iteratively for 100 epochs until the change of loss function is less than 1e-5. We plot the training loss over time in Figure \ref{book_epoch}, \ref{movie_epoch} and \ref{music_epoch} for three domain pairs. The key observation is that, DDTCDR starts with relatively higher loss due to the noisy initialization. As times goes by, DDTCDR manages to stabilize quickly and significantly outperforms NCF only after 10 epochs.

\begin{figure}[h]
\centering
\includegraphics[width=0.4\textwidth]{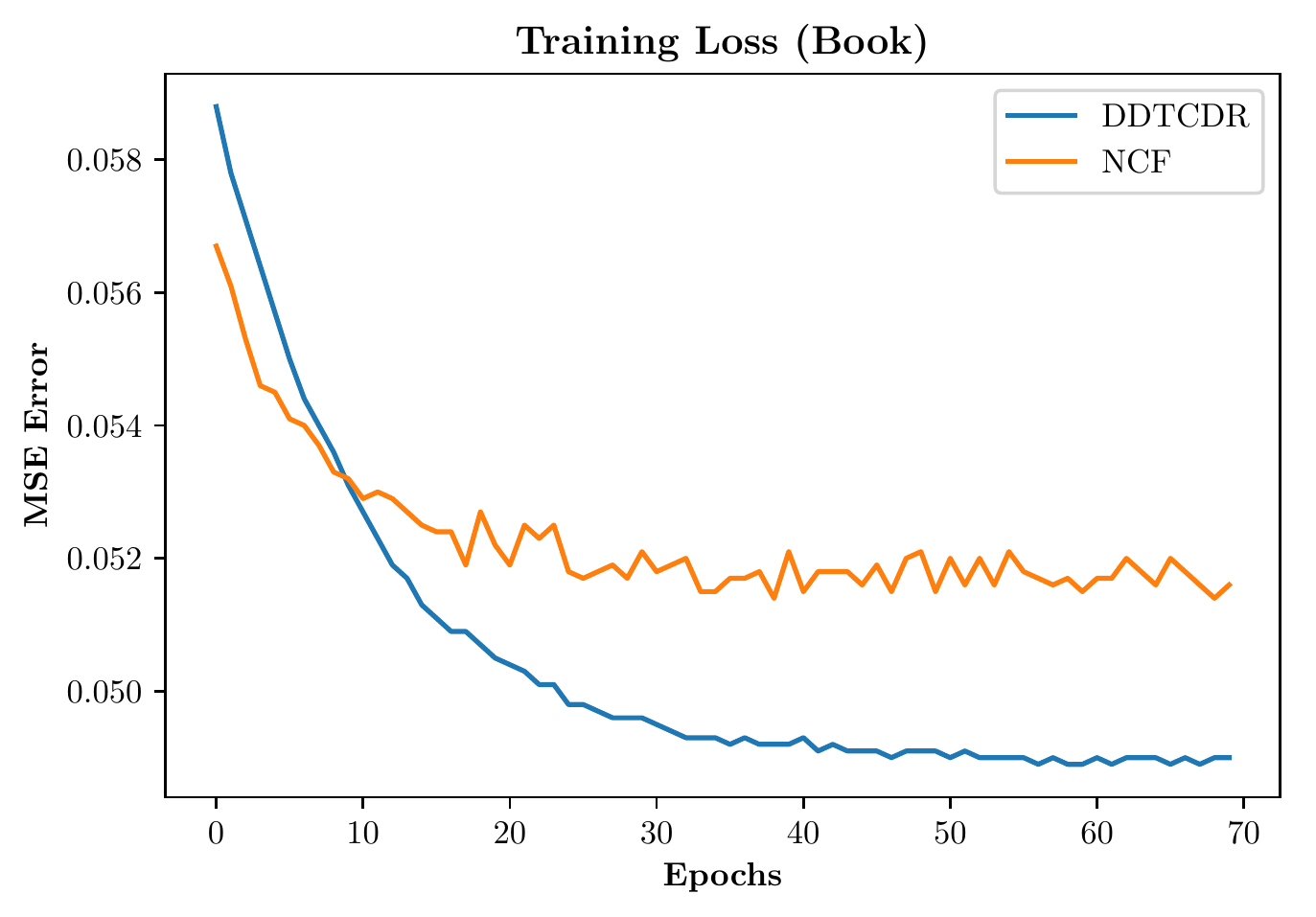}
\caption{Book Domain Epoch-Loss}
\label{book_epoch}
\end{figure}

\begin{figure}[h]
\centering
\includegraphics[width=0.4\textwidth]{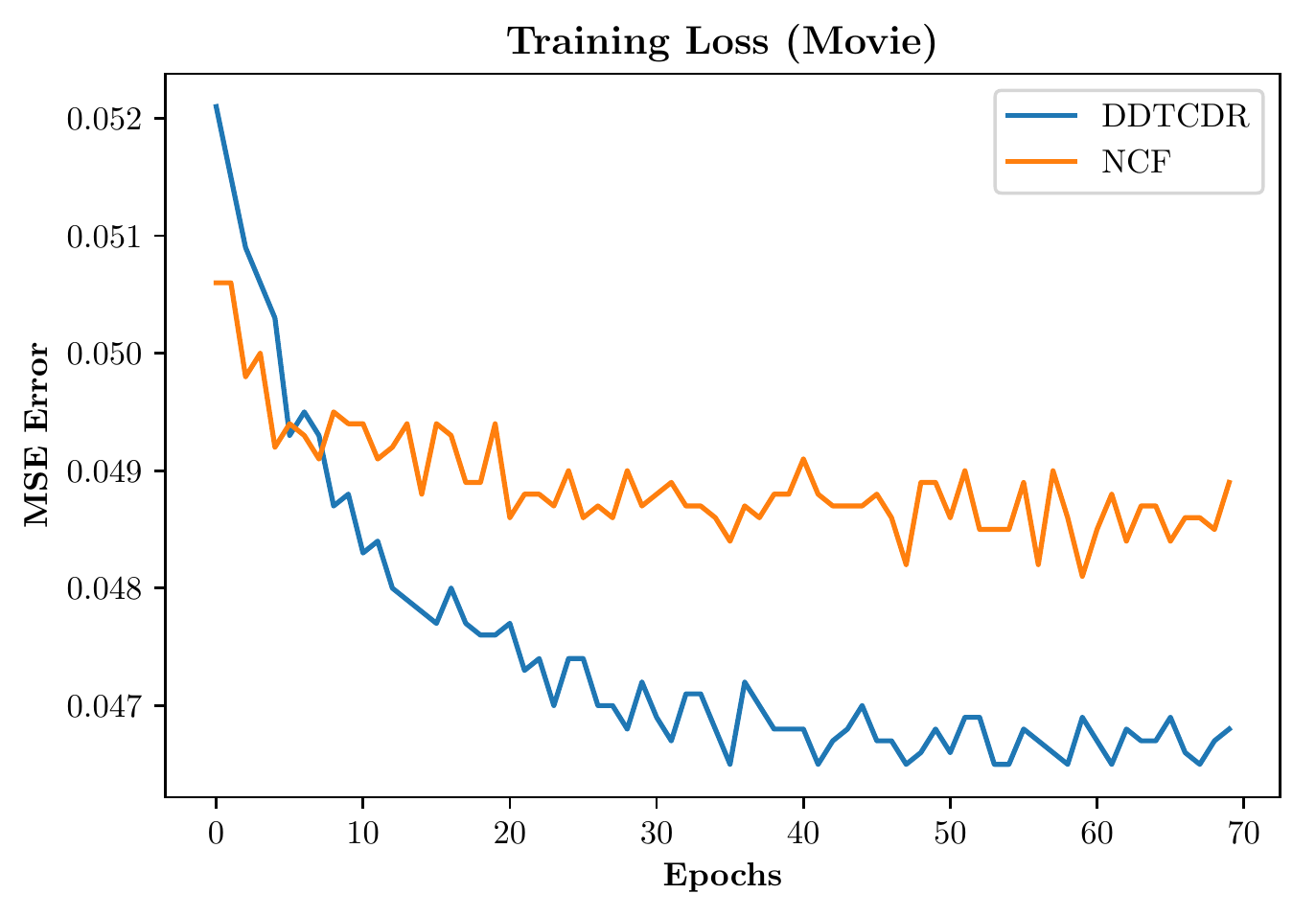}
\caption{Movie Domain Epoch-Loss}
\label{movie_epoch}
\end{figure}

\begin{figure}[h]
\centering
\includegraphics[width=0.4\textwidth]{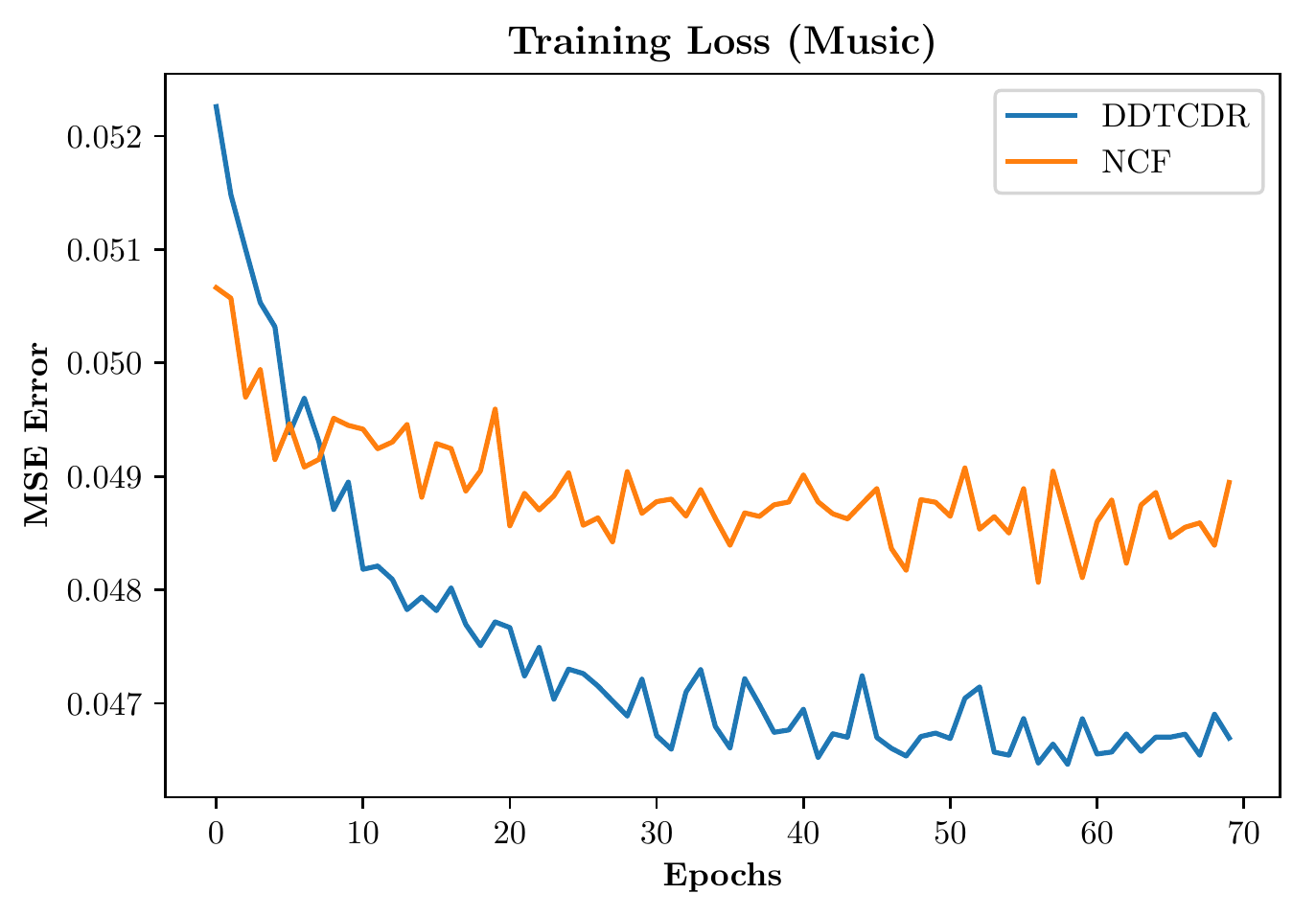}
\caption{Music Domain Epoch-Loss}
\label{music_epoch}
\end{figure}

\subsection{Sensitivity}
To validate that the improvement gained in our model is not sensitive to the specific experimental setting, we conduct multiple experiments to examine the change of recommendation performance corresponding to different hyperparameter settings. As Figure \ref{book_alpha}, \ref{movie_alpha} and \ref{music_alpha} show, we could observe certain amount of performance fluctuation with respect to transfer rate $\alpha$. Note that the cross-domain recommendation results (when $\alpha$ take positive and small value) are consistently better than that of single-domain recommendations (when $\alpha = 0$). Also, we verify the recommendation performance using different types of autoencoders, including AE \cite{sutskever2014sequence}, VAE \cite{kingma2013auto}, AAE \cite{makhzani2015adversarial}, WAE \cite{tolstikhin2017wasserstein} and HVAE \cite{davidson2018hyperspherical}, and the improvement of recommendation performance is consistent across these settings as shown in Table \ref{autoencoder}. The choice of particular autoencoder is not relevant to the recommendation performance.

\begin{figure}[h]
\centering
\includegraphics[width=0.35\textwidth]{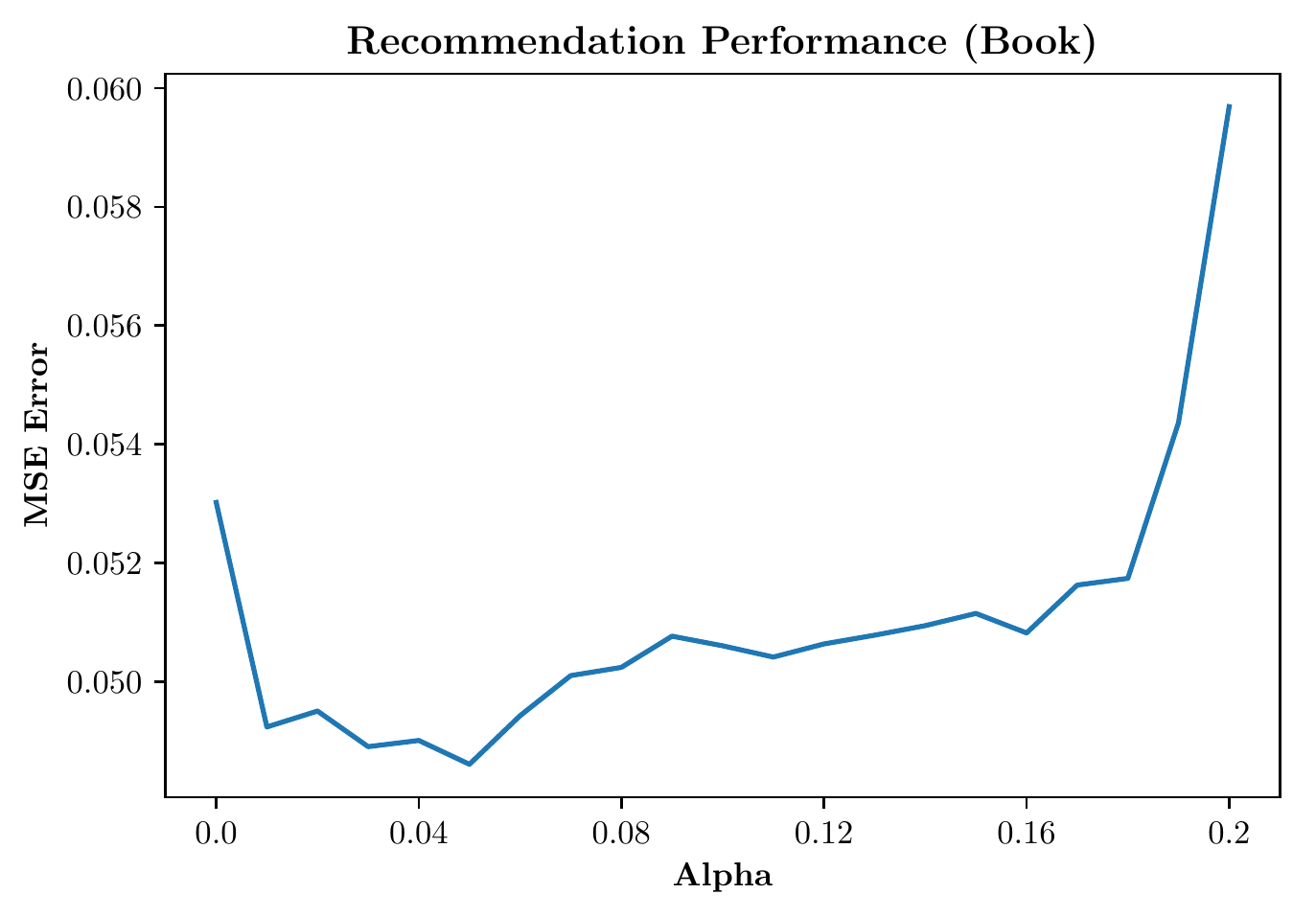}
\caption{Recommendation Performance in Book Domain with Different Alpha Values}
\label{book_alpha}
\end{figure}

\begin{figure}[h]
\centering
\includegraphics[width=0.35\textwidth]{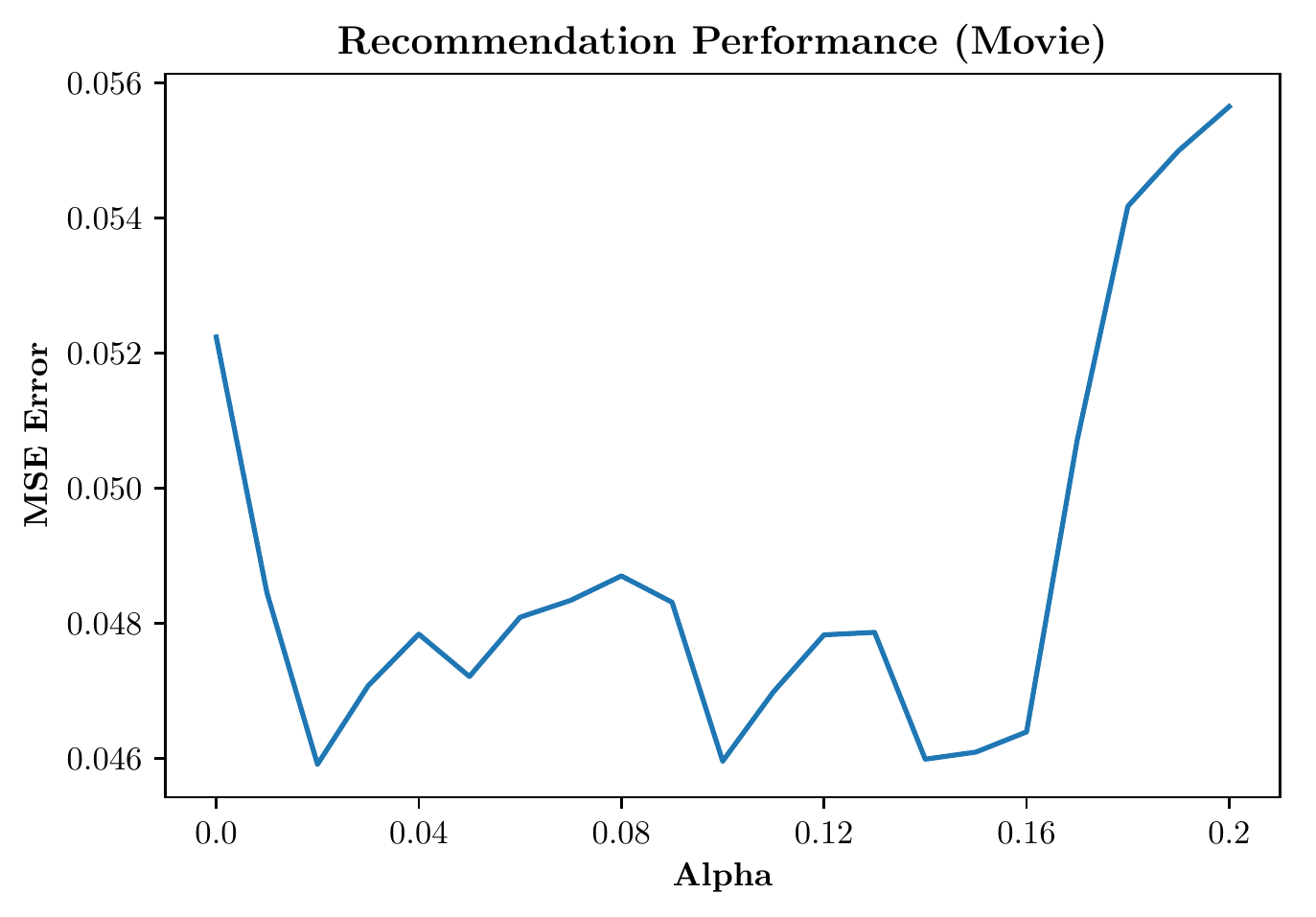}
\caption{Recommendation Performance in Movie Domain with Different Alpha Values}
\label{movie_alpha}
\end{figure}

\begin{figure}[h]
\centering
\includegraphics[width=0.35\textwidth]{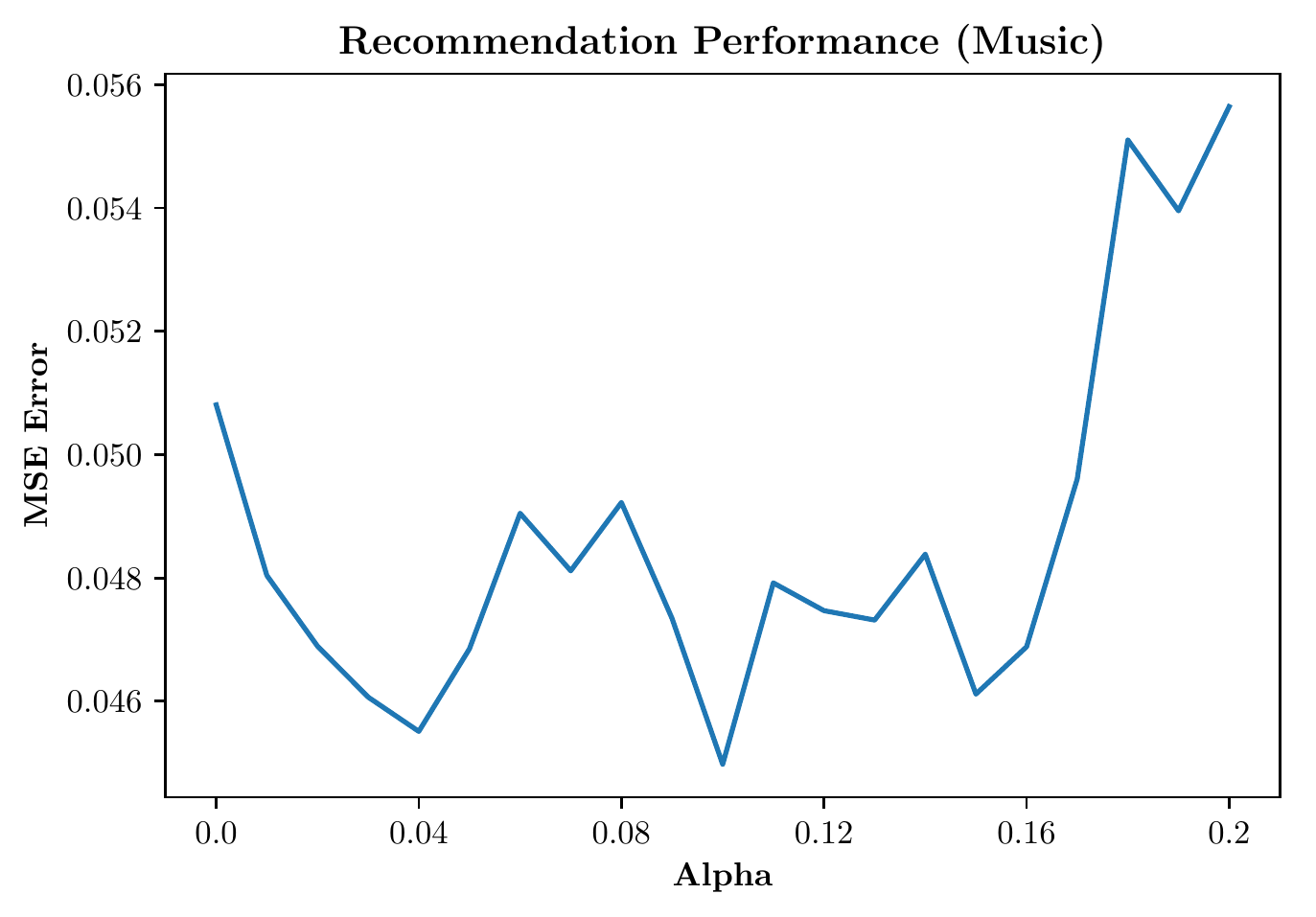}
\caption{Recommendation Performance in Music Domain with Different Alpha Values}
\label{music_alpha}
\end{figure}

\begin{table}
\centering
\begin{tabular}{|c|cc|cc|} \hline
\multirow{2}{*}{Autoencoder} & \multicolumn{4}{c|}{Book $\Longleftrightarrow$ Movie} \\ \cline{2-5}
               &  RMSE & MAE & RMSE & MAE \\ \hline
AE  & 0.2213 & 0.1708 & 0.2213 & 0.1714  \\ \hline
VAE &  0.2240 &  0.1704 &  0.2196 & 0.1707 \\ \hline
AAE & 0.2236  & 0.1729  & 0.2195  &  0.1715 \\ \hline
WAE & 0.2236 & 0.1739  & 0.2202 & 0.1739 \\ \hline
HVAE & 0.2220 & 0.1717 & 0.2186 &  0.1704 \\ \hline
\end{tabular}
\newline
\caption{Comparison of Autoencoder Settings: Differences are not statistically significant}
\label{autoencoder}
\end{table}

\section{Conclusion}
In this paper, we propose a novel dual transfer learning based model that significantly improves recommendation performance across different domains. We accomplish this by transferring latent information from one domain to the other through embeddings and iteratively going through the transfer learning loop until the models stabilize for both domains. We also prove that this convergence is guaranteed under certain conditions and empirically validate the hypothesis for our model across different experimental settings. 

Note that, the proposed approach provides several benefits, including that it (a) transfers information about latent interactions instead of explicit features from the source domain to the target domain; (b) utilizes the dual transfer learning mechanism to enable the bidirectional training that improves performance measures for both domains simultaneously; (c) learns the latent orthogonal mapping function across two domains, that (i) preserves similarity of user preferences and thus enables proper transfer learning and (ii) computes the inverse mapping function efficiently.

As the future work, we plan to extend dual learning mechanism to multiple domains by simultaneously improving performance across all domains instead of domain pairs. We also plan to extend the convergence proposition to more general settings.

%
\bibliographystyle{ACM-Reference-Format}
\bibliography{sigproc}

\end{document}